\documentclass[11pt]{article}
\usepackage{amsmath, amsthm}
\usepackage{amsmath,amssymb,latexsym}
\newtheorem{theorem}{Theorem}
\newtheorem{lemma}{Lemma}

\usepackage{graphicx}
\usepackage{psfrag}
\usepackage{subfigure}
\usepackage{appendix}
\usepackage{amsfonts}
\usepackage{epstopdf}
\newcommand{\bbee}{\begin{equation}}
\newcommand{\eeee}{\end{equation}}
\newcommand{\bbaa}{\begin{array}}
\newcommand{\eeaa}{\end{array}}
\usepackage{ifpdf}
\usepackage{cite}
\usepackage{eqparbox}
\usepackage{fixltx2e}
\usepackage{stfloats}
\usepackage{mdwmath}
\usepackage{mdwtab}
\usepackage{array}
\usepackage{comment}
\begin{document}
\thispagestyle{empty}
\setcounter{page}{1}
\setlength{\baselineskip}{1.5\baselineskip}
\title{Universal Quantization for Separate Encodings and Joint Decoding of Correlated Sources \thanks{This research was supported by the Israeli Science Foundation (ISF),
grant no.\ 208/08.}\footnotetext{This paper was presented in part at 2014 IEEE International Symposium on Information Theory (ISIT).}}
\author{Avraham Reani and Neri Merhav\\
Department of Electrical Engineering\\
Technion - Israel Institute of Technology\\
Technion City, Haifa 32000, Israel\\
Emails: [avire@tx, merhav@ee].technion.ac.il}\maketitle
\begin{abstract}
We consider the multi-user lossy source-coding problem for continuous alphabet sources. In a previous work, Ziv proposed a single-user universal coding scheme which uses uniform quantization with dither, followed by a lossless source encoder (entropy coder). In this paper, we generalize Ziv's scheme to the multi-user setting. For this generalized universal scheme, upper bounds are derived on the redundancies, defined as the differences between the actual rates and the closest corresponding rates on the boundary of the rate region. It is shown that this scheme can achieve redundancies of no more than 0.754 bits per sample for each user. These bounds are obtained without knowledge of the multi-user rate region, which is an open problem in general. As a direct consequence of these results, inner and outer bounds on the rate-distortion achievable region are obtained.
\\\\{\bf Index Terms: Multi-terminal source coding, Dithered quantization, Universal source coding, scalar quantization, Slepian-Wolf coding.}
\end{abstract}
\clearpage
\section{Introduction}
Consider the case where two correlated sources are observed separately
by two non-cooperative encoders which communicate with one decoder. The
decoder needs to reconstruct both sources and the distortions between the reconstructions and the corresponding sources should not exceed some given values. The general version of this problem has remained open for several decades, even under the assumption of memoryless sources. However, many special cases have been solved. When no distortion is allowed, this is the problem considered by Slepian and Wolf \cite{Slepian_Wolf}. Their well-known result states that two discrete
sources $X_1$ and $X_2$ can be losslessly reproduced if and only if 
\begin{subequations}
	\begin{eqnarray}
 {R}_1&\geq& H(X_1|X_2),\\
 {R}_2&\geq& H(X_2|X_1),\\
 {R}_1+{R}_2&\geq& H(X_1,X_2)
 	\end{eqnarray}
 \label{slepian_wolf_eq}
\end{subequations}
where $R_1$ is the rate of the encoder observing $X_1$ and $R_2$ is the
rate of the encoder observing $X_2$. Returning to the lossy case, the setting in which one of the variables is known to the decoder, is the original Wyner-Ziv problem \cite{Wyner_Ziv}. This setting was generalized to continuous alphabet sources by Wyner \cite{Wyner}.
Other examples include the source coding problem with side information of Ahlswede-K\"{o}rner \cite{Ahlswede_Korner}, where an arbitrary distortion is allowed for one of the sources and the other source should be reconstructed losslessly. Berger
and Yeung \cite{Berger_Yeung} considered a setting where one of the sources is to be perfectly
reconstructed and the other source should be reconstructed with a distortion
constraint (their setting subsumes all previous examples). Zamir and Berger \cite{Zamir_Berger} characterized the rate-distortion region in the high-SNR limit. Wagner and Anantharan \cite{Wagner_Anantharam} presented a new outer bound which is better than the previous outer bounds in the literature.

Recent results for specific sources and distortion measures include the works of Wagner, Tavildar, and Viswanath \cite{Wagner_Tavildar_Viswanath}, who determined the rate region for the quadratic Gaussian multiterminal source coding problem, by showing that the Berger-Tung \cite{Berger_Tung} inner bound is tight. In addition, a characterization of the rate region
under logarithmic loss was given by Courtade and Weissman \cite{Courtade_Weissman}.
Finally, a version of this problem, where both users and the decoder must operate with
zero-delay, was considered by Kaspi and Merhav \cite{Kaspi_Merhav}, who characterized the rate region in this case.

In \cite{Ziv}, Ziv presented a universal coding scheme for the single-user case. This scheme is composed of a uniform, one-dimensional quantizer with dither, followed by a noiseless variable-rate encoder (entropy encoder). He showed that this scheme yields a rate that is, for every positive integer $n$, no more than $0.754$ bits per sample higher than the best possible rate associated with the optimal $n$-dimensional quantizer. This result was later revisited and further developed by Zamir and Feder \cite{Zamir_Feder1}, \cite{Zamir_Feder2}, who also gave a redundancy upper bound which depends on the source distribution. However, their derivation of the global upper bound relies on the known formula of the single-user rate-distortion function. In addition, a dithered scheme for the multi-user setting, which is similarly to the scheme in this paper, was given in \cite{Zamir_Berger}. Redundancy upper bounds can be derived by bounding the difference between the dithered scheme rate region and the outer bound on the multi-user rate region given in \cite{Zamir_Berger}. These bounds depend on the divergence between the source distribution and a Gaussian distribution. As a result, they are not uniformly bounded (for every source distribution) in contrast to the bound of Ziv and the bounds presented in this paper. In addition, only the redundancy of the sum of the rates can be upper bounded using the methods of \cite{Zamir_Berger}.

In this paper, we investigate a generalized scheme for the multi-user setting. In this scheme, each user uses dithered quantizer followed by universal Slepian-Wolf encoder. We show that the rates achieved by this scheme are no more than 0.754 bits per sample away from the boundary of the achievable rate region, for each user. This is done regardless of the characterization of the achievable region, which is, as mentioned before, unknown in general. As a direct consequence of these results, inner and outer bounds on the achievable region are obtained. Finally, similarly to the results of \cite{Ziv}, it is straightforward to show that using multi-dimensional lattice quantizers instead of scalar ones would decrease the redundancy to about 0.5 bits per sample for high lattice dimension.

The remainder of this paper is organized as follows. In Section 2, we present the problem formulation and give basic results regarding the performance of the dithered scheme. In Section 3, we revisit the redundancy upper bound of \cite{Ziv}. In Section 4, we enhance the results of Section 2 by adding an estimation stage to the dithered scheme. We conclude this work in Section 5.
\section{Problem Formulation and Basic Results}
Throughout the paper, random variables will be denoted by capital
letters and their alphabets will be denoted by calligraphic letters. Random
vectors (all of length $n$) will be denoted by capital letters in the bold face font.

In this section, we present the multi-user setting we deal with and describe the dithered coding scheme we use. Then, we give upper bounds on the performance of this scheme, compared to the boundary of the optimal rate region.

We begin with defining the multi-user rate region. Let $(X_1,X_2)$ be a continuous alphabet memoryless source, characterized by the joint probability density $P_{X_1X_2}$. We assume that $P_{X_1X_2}$ has bounded support, i.e., there exists $A\in \mathbb{R}^+$ such that $P_{X_1X_2}(x_1,x_2)=0$ if $(x_1,x_2)\notin [-A,A]\times [-A,A]$. The reason for this assumption will be explained later. A rate pair $(R^{*}_1,R^{*}_2)$ is said to be $(D_1,D_2)$-achievable under the mean-square error distortion measure with respect to $(X_1,X_2)$, if for every $\delta>0$ and sufficiently large $n$, there exists a code of block length $n$ consisting of two encoders ${f}_1$, ${f}_2$
\begin{eqnarray}
{f}_1:[-A,A]^n\rightarrow I_{M_1},&{f}_2:[-A,A]^n\rightarrow I_{M_2}
\label{MT1}
\end{eqnarray}
and a decoder $g$
\begin{eqnarray}
g:I_{M_1} \times I_{M_2}\rightarrow [-A,A]^n\times [-A,A]^n
\label{MT2}
\end{eqnarray}
such that
\begin{eqnarray}
\displaystyle\frac{1}{n} \mathbb{E}||\bold{X}_1-\bold{\hat{X}}_1||^2\leq D_1+\delta,&\displaystyle\frac{1}{n} \mathbb{E}||\bold{X}_2-\bold{\hat{X}}_2||^2\leq D_2+\delta
\end{eqnarray}
and
\begin{eqnarray}
\displaystyle\frac{1}{n}\log M_1\leq R^{*}_1+\delta,&\displaystyle\frac{1}{n}\log M_2\leq R^{*}_2+\delta,
\end{eqnarray}
where $I_{M_i}\triangleq\{1,2,\ldots,M_i\}$, $i\in\{1,2\}$. The set of $(D_1,D_2)$-achievable rate pairs, is denoted by ${\cal{R}}^{*}(D_1,D_2)$.

Our scheme works as follows. We have two encoders $\tilde{f}_1$, $\tilde{f}_2$:
\begin{eqnarray}
\tilde{f}_1:[-A,A]^n \times [-\sqrt{3D_1},\sqrt{3D_1}]\rightarrow I_{\tilde{M}_1},&\tilde{f}_2:[-A,A]^n\times [-\sqrt{3D_2},\sqrt{3D_2}]\rightarrow I_{\tilde{M}_2}
\end{eqnarray}
and a decoder $\tilde{g}$
\begin{eqnarray}
\tilde{g}:I_{\tilde{M}_1} \times I_{\tilde{M}_2}\times [-\sqrt{3D_1},\sqrt{3D_1}]\times[-\sqrt{3D_2},\sqrt{3D_2}]\rightarrow [-A,A]^n\times [-A,A]^n.
\end{eqnarray}
Each encoder $\tilde{f}_i$, $i\in\{1,2\}$, uses a one-dimensional uniform quantizer ${\cal{Q}}_i$, ${\cal{Q}}_i:{\mathbb{R}}\rightarrow \{0,\pm 2\sqrt{3D_i},\pm 2\cdot 2\sqrt{3D_i},\ldots\}$ and a dither random variable (RV) ${Z_i}$, uniformly distributed over $[-\sqrt{3D_i},\sqrt{3D_i}]$, to produce 
$\bold{\cal{Q}}_i(\bold{X}_i+\bold{Z}_i)\triangleq [{\cal{Q}}_i(X_{i,1}+Z_{i}),{\cal{Q}}_i(X_{i,2}+Z_{i}),\ldots,{\cal{Q}}_i(X_{i,n}+Z_{i})]$, where $\bold{Z}_i$ denotes a vector of dimension $n$ composed of $n$ repetitions of the same realization of ${Z_i}$. For convenience,  the random variable ${\cal{Q}}_i({X}_{i}+{Z}_i)$ and the random vector $\bold{\cal{Q}}_i(\bold{X}_{i}+\bold{Z}_i)$ will be denoted by $Y_{i}$ and $\bold{Y}_i$, respectively. The dither RV's, $Z_1$ and $Z_2$, are available to the respective encoders and to the decoder and are independent.
As is shown in \cite[Lemma 1]{Ziv},
\begin{eqnarray}
\displaystyle{\mathbb{E}}\left[||Y_i-{Z}_i-{X}_i||^2|X_i\right]=D_i,&i\in\{1,2\}
\end{eqnarray}
where the expectation is taken over ${Z_i}$. Notice that the distortion is $D_i$ independently of $X_i$ and therefore the total distortion is also $D_i$. After the quantization stage, the two encoders perform Slepian-Wolf encoding with a rate pair $(R_1,R_2)$, for lossless compression of $\bold{Y}_1$ and $\bold{Y}_2$. Complying with Eq. (\ref{slepian_wolf_eq}), the rate pair $(R_1=\log\tilde{M}_1,R_2=\log\tilde{M}_2)$ satisfies 
\begin{figure}[t]
	\centering		
	\includegraphics[width=6in, height=1.8 in]{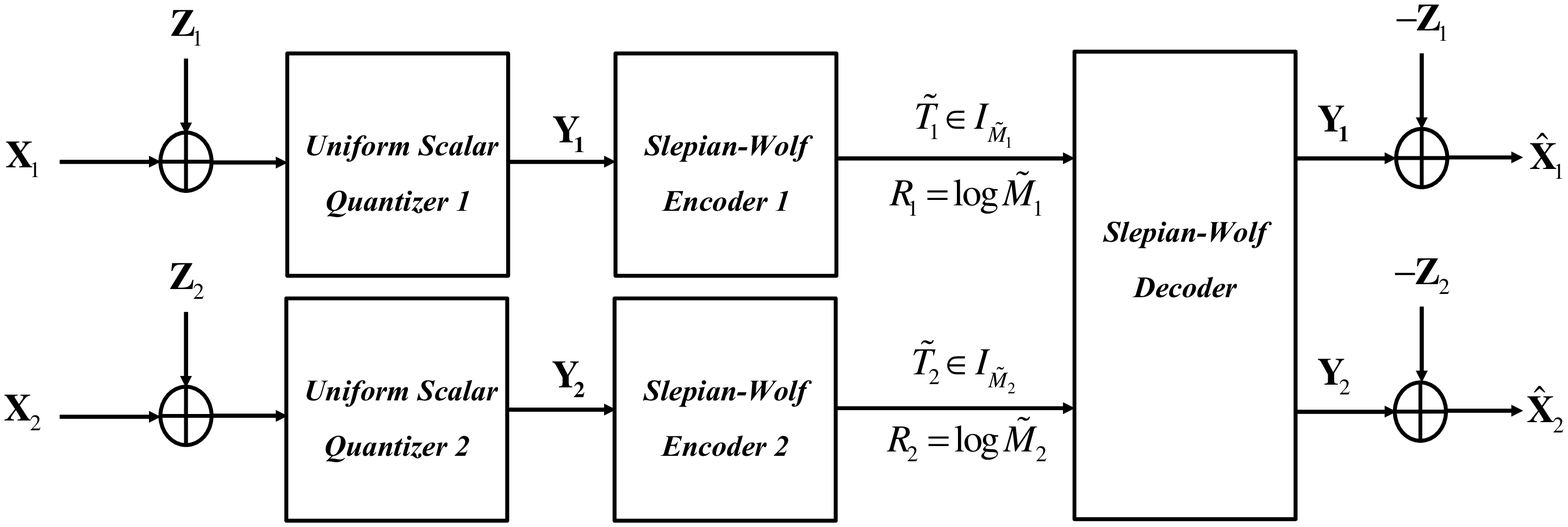}
		\caption { {The dithered coding scheme}} 
		\label{figure1}
\end{figure}
\begin{subequations}
\begin{eqnarray}
R_1&\geq&\displaystyle H\left({Y}_1|{Y}_2,Z_1,Z_2\right),\\
R_2&\geq&\displaystyle H\left({Y}_2|{Y}_1,Z_1,Z_2\right),\\
R_1+R_2&\geq& H\left({Y}_1,{Y}_2|Z_1,Z_2\right)
\end{eqnarray}
\label{RR}
\end{subequations}
where we used the following, for every value of $n$
\begin{subequations}
\begin{eqnarray}
\label{memoryless}
\displaystyle\frac{1}{n} H\left(\bold{Y}_1|\bold{Y}_2,Z_1,Z_2\right)&=&H\left(Y_1|Y_2,Z_1,Z_2\right),\\
\displaystyle\frac{1}{n} H\left(\bold{Y}_2|\bold{Y}_1,Z_1,Z_2\right)&=&H\left(Y_2|Y_1,Z_1,Z_2\right),\\
\displaystyle\frac{1}{n} H\left(\bold{Y}_1,\bold{Y}_2|Z_1,Z_2\right)&=&H\left(Y_1,Y_2|Z_1,Z_2\right).
\end{eqnarray}
\end{subequations}
To see why (\ref{memoryless}) is true, consider the following chain
\begin{eqnarray}
\nonumber\displaystyle\frac{1}{n} H\left(\bold{Y}_1|\bold{Y}_2,Z_1,Z_2\right)&=&\displaystyle\frac{1}{n} \sum_{i=1}^{n} H\left({Y}_{1,i}|{Y}_{1,1},{Y}_{1,2},\ldots,Y_{1,i-1},\bold{Y}_2,Z_1,Z_2\right)\\
\nonumber&=&\displaystyle\frac{1}{n} \sum_{i=1}^{n} H\left(Y_{1,i}|Y_{2,i},Z_1,Z_2\right)\\
&=&H\left(Y_1|Y_2,Z_1,Z_2\right)
\end{eqnarray}
where the second equality stems from the fact that $Y_1$ and $Y_2$ are memoryless given $Z_1$ and $Z_2$ and the third equality stems from the stationarity of the source. The same can be done for $ H\left(\bold{Y}_2|\bold{Y}_1,Z_1,Z_2\right)$ and $H\left(\bold{Y}_2,\bold{Y}_2|Z_1,Z_2\right)$.

The rate region of Eq. (\ref{RR}) is achievable for $n$ sufficiently large and it is denoted by ${\cal{R}}(D_1,D_2)$. The interesting range of $R_1$ is $\overline{{\cal{R}}}_1(D_1,D_2)\triangleq\left[\displaystyle H\left({Y}_1|{Y}_2,Z_1,Z_2\right),\displaystyle H\left({Y}_1|Z_1\right)\right]$ since higher rate can always be reduced to this range. The same is true for $R_2$. The universal decoder first decodes $\bold{Y}_1$ and $\bold{Y}_2$ (correctly with high probability), and then subtracts the corresponding dithers to obtain the reconstruction vectors $\bold{\hat{X}}_1$, $\bold{\hat{X}}_2$:
\begin{eqnarray}  
\bold{\hat{X}}_i=\bold{Y}_i-\bold{Z}_i.
\end{eqnarray}
The universal Slepian-Wolf decoder is described in Appendix A. The dithered coding scheme is presented in Fig. 1.\\
{\it{Remark}}. The Slepian-Wolf mechanism can be applied, in general, to sources with countably-infinite alphabets. However, a universal Slepian-Wolf scheme for such sources is not known. Trying to preserve universality in the case of infinite alphabets would require the assignment of infinite number of sequences into bins. Thus, even the codebook generation  does not seem to be feasible in this case. This is not surprising, considering the fact that even in the single-user case, diminishing redundancy cannot be achieved for universal lossless coding of sources with infinite alphabets (see, e.g., \cite{Kieffer}). Therefore, for the sake of universality, we assumed that the source alphabets have bounded supports so the outputs of the quantizers have finite alphabets. From the above, this assumption is also needed for the original single-user scheme of Ziv \cite{Ziv}. The inner and outer bounds on the achievable rate-distortion region, which are obtained as a direct consequence of Theorems 1-4 below, are also valid, of course, for sources with unbounded support, as they do not depend on the universality. 

We begin with a simple result.
\begin{theorem}
For any rate pair $(R^{*}_1,R^{*}_2)$ on the boundary of ${\cal{R}}^{*}(D_1,D_2)$ and any rate pair $(R_1,R_2)$ on the boundary of ${\cal{R}}(D_1,D_2)$, with $R_1\in \overline{{\cal{R}}}_1(D_1,D_2)$, we have
\begin{eqnarray}
R_1+R_2\leq R^{*}_1+R^{*}_2+2c
\end{eqnarray}
where $c=0.754$ bits/sample.\\
Moreover, for any $R^{*}_1\in\overline{\cal{R}}_1(D_1,D_2)$, there exists a rate pair $(R_1,R_2)\in{\cal{R}}(D_1,D_2)$ such that
\begin{eqnarray}
\nonumber R_1&=&R^{*}_1\\
R_2&\leq& R^{*}_2+2c.
\end{eqnarray}
\label{the222}
\end{theorem}
\begin{proof}[Proof of Theorem \ref{the222}]
We have
\begin{eqnarray}
&&\nonumber \displaystyle\frac{1}{n}H\left(\bold{Y}_1,\bold{Y}_2|Z_1,Z_2\right)\\
\nonumber &\leq&\displaystyle\frac{1}{n} H\left(\bold{Y}_1,\bold{Y}_2,T_1,T_2|Z_1,Z_2\right)\\
\nonumber &\leq&\displaystyle\frac{1}{n}H(T_1,T_2)\nonumber+\displaystyle\frac{1}{n}H\left(\bold{Y}_1,\bold{Y}_2|T_1,T_2,Z_1,Z_2\right)\\
\nonumber &\leq& R^{*}_1+R^{*}_2+\displaystyle\frac{1}{n}H\left(\bold{Y}_1,\bold{Y}_2|T_1,T_2,Z_1,Z_2\right)\\
\nonumber &\leq& R^{*}_1+R^{*}_2+\displaystyle\frac{1}{n}H\left(\bold{Y}_1,\bold{Y}_2|g(T_1,T_2),Z_1,Z_2\right)\\
\nonumber &=& R^{*}_1+R^{*}_2\nonumber+\displaystyle\frac{1}{n}H\left(\bold{Y}_1,\bold{Y}_2|\bold{\hat{X}}^{opt}_1,\bold{\hat{X}}^{opt}_2,Z_1,Z_2\right)\\
\nonumber &\leq& R^{*}_1+R^{*}_2+\displaystyle\frac{1}{n}H\left(\bold{Y}_1|\bold{\hat{X}}^{opt}_1,Z_1\right)\\
\nonumber &&+\displaystyle\frac{1}{n}H\left(\bold{Y}_2|\bold{\hat{X}}^{opt}_2,Z_2\right)\\
&\leq&R^{*}_1+R^{*}_2+2c
\end{eqnarray}
where $T_1\in I_{M_1}$, $T_2\in I_{M_2}$ are the outputs of the optimal encoders $f_1$, $f_2$, respectively, $\left(\bold{\hat{X}}^{opt}_1,\bold{\hat{X}}^{opt}_2\right)\triangleq g(T_1,T_2)$ are the outputs of the optimal decoder $g$, and $(R^{*}_1,R^{*}_2)\in{\cal{R}}^{*}(D_1,D_2)$. The last inequality can be obtained in the same way as in \cite{Ziv}. The left-hand side is achievable for sufficiently large $n$. Therefore, for any rate pair $(R_1,R_2)\in{\cal{R}}(D_1,D_2)$, which lies on the straight line $R_1+R_2=\displaystyle H\left({Y}_1,{Y}_2|Z_1,Z_2\right)$, we have
\begin{eqnarray}
R_1+R_2&\leq&R^{*}_1+R^{*}_2+2c
\label{eqthe1}
\end{eqnarray}
Moreover, if $R^{*}_1\in \overline{{\cal{R}}}_1(D_1,D_2)$, we can always take $R_1=R^{*}_1$ and obtain:
\begin{eqnarray}
R_2&\leq&R^{*}_2+ 2c 
\end{eqnarray}
The same can be done, of course, when the roles of the two users are interchanged. This completes the proof.
\end{proof}
The following theorem suggests another result regarding the relation between the boundary of ${\cal{R}}(D_1,D_2)$ and that of ${\cal{R}}^{*}(D_1,D_2)$.
\begin{theorem}
For any rate pair $(R_1,R_2)$ on the boundary of ${\cal{R}}(D_1,D_2)$, with $R_1\in \overline{{\cal{R}}}_1(D_1,D_2)$,
there exists a rate pair $(R^{*}_1,R^{*}_2)\in{\cal{R}}^{*}(D_1,D_2)$ such that:
\begin{eqnarray}
\nonumber R_1&\leq& R^{*}_1+c\\
 R_2&\leq& R^{*}_2+c
\label{the22equation}
\end{eqnarray}
\label{thestrong}
\end{theorem}
Notice that Theorems 1 and 2 also provide outer bounds on ${\cal{R}}^{*}(D_1,D_2)$. Theorem 1 asserts that the straight line $R_1+R_2=\displaystyle H\left({Y}_1,{Y}_2|Z_1,Z_2\right)-2c$ defines an outer bound for ${\cal{R}}^{*}(D_1,D_2)$. In addition, Theorem 2 bounds the distance between the boundary of ${\cal{R}}(D_1,D_2)$ and that of ${\cal{R}}^{*}(D_1,D_2)$ in each coordinate. The boundary of ${\cal{R}}(D_1,D_2)$ is, of course, an inner bound on ${\cal{R}}^{*}(D_1,D_2)$.

Before proving Theorem 2, we first prove a simple auxiliary result regarding the source-coding problem where side information is available only to the encoders but not to the decoder. The setting is as follows. A rate pair $(R_1,R_2)$ is achievable for a memoryless source $({\cal{Y}}_1,{\cal{Y}}_2,P_{Y_1,Y_2})$ and some side information $S\in{\cal{S}}$ which depends statistically on $(\bold{Y}_1,\bold{Y}_2)$ through the joint probability distributions $P_{\bold{Y}_1,\bold{Y}_2,S}$, if for any $\delta>0$ and sufficiently large $n$, there exists a block code of length $n$ consisting of two encoders ${f}_1$, ${f}_2$
\begin{eqnarray}
{f}_1:{\cal{Y}}_1^n\times {\cal{S}}\rightarrow I_{M_1},&{f}_2:{\cal{Y}}_2^n\times {\cal{S}}\rightarrow I_{M_2}
\label{MT11}
\end{eqnarray}
and a decoder $g$
\begin{eqnarray}
g:I_{M_1} \times I_{M_2}\rightarrow {{\cal{Y}}}_1^n\times {{\cal{Y}}}_2^n
\label{MT22}
\end{eqnarray}
such that
\begin{eqnarray}
\Pr\{g\left(f_1(\bold{Y}_1,S),f_2(\bold{Y}_2,S)\right)\neq\left(\bold{Y}_1,\bold{Y}_2\right)\}\leq \delta
\end{eqnarray}
and
\begin{eqnarray}
\displaystyle\frac{1}{n}\log M_1\leq R_1+\delta,&\displaystyle\frac{1}{n}\log M_2\leq R_2+\delta
\end{eqnarray}
The set of achievable rate pairs is denoted by $\tilde{\cal{R}}$. The regular Slepian-Wolf region (without side information) is denoted by ${\cal{R}}_{SW}$. Obviously,  ${\cal{R}}_{SW}\subseteq\tilde{\cal{R}}$.
We have the following lemma.
\begin{lemma}
Any rate pair $(\tilde{R}_1,\tilde{R}_2)\in \tilde{\cal{R}}$ must satisfy the following constraint:
\begin{eqnarray}
\tilde{R}_1+\tilde{R}_2\geq H(Y_1,Y_2).
\end{eqnarray}
Therefore, side information available only to the encoders cannot improve the performance if $\tilde{R}_1\in\left[ H(Y_1|Y_2),H(Y_1)\right]$ or $\tilde{R}_2\in\left[ H(Y_2|Y_1),H(Y_2)\right]$.
\label{lem11}
\end{lemma}
\begin{proof}[Proof of Lemma \ref{lem11}]
The proof follows directly from the fact that even one encoder, which has access to $(\bold{Y}_1,\bold{Y}_2,S)$, cannot do better than $H(\bold{Y}_1,\bold{Y}_2)$, when the side information $S$ is not available to the decoder.
\end{proof}
The generalization of Lemma 1 to our case where, in addition, a dither is available to the encoders and decoder, is straightforward.
We can now prove Theorem \ref{thestrong}.
\begin{proof}[Proof of Theorem \ref{thestrong}]
Assume that the optimal code $(f_1,f_2,g)$, which achieves the rate pair $(R^{*}_1,R^{*}_2)$, is known, and that the encoders of the dithered scheme, which transmit $\bold{Y}_1$, $\bold{Y}_2$ at rates $(R_1,R_2)$ to the decoder, have access to $f_1(\bold{X}_1)$, $f_2(\bold{X}_2)$ as side information. According to Lemma 1, this side information does not change the fact that any rate pair $(R_1,R_2)\in{\cal{R}}(D_1,D_2)$ must satisfy $R_1+R_2\geq\displaystyle H({Y}_1,{Y}_2|Z_1,Z_2)$.
Consider the following auxiliary coding scheme:
User $i$ compresses $T_i=f_i(\bold{X}_i)$ using $nR^{*}_i$ bits, $i\in\{1,2\}$. Then, the first user uses Slepian-Wolf coding to compress $\bold{Y}_1$ given $\{T_1,T_2,Z_1\}$ into $H(\bold{Y}_1|T_1,T_2,Z_1)$ bits. The second user uses Slepian-Wolf coding to compress $\bold{Y}_2$ given $\{\bold{Y}_1,T_1,T_2,Z_1,Z_2\}$ into $H(\bold{Y}_2|\bold{Y}_1,T_1,T_2,Z_1,Z_2)$ bits. 
The decoder, which has access to $\{T_1,T_2,Z_1,Z_2\}$, first decodes $\bold{Y}_1$, using $\{T_1,T_2,Z_1\}$. Then, it decodes $\bold{Y}_2$ using $\{\bold{Y}_1,T_1,T_2,Z_1,Z_2\}$. The rate pair of this scheme, $(R_1,R_2)$, satisfies
\begin{eqnarray}
\nonumber R_1&=&R^{*}_1+\displaystyle\frac{1}{n}H\left(\bold{Y}_1|T_1,T_2,Z_1\right)\\
\nonumber &\leq&R^{*}_1+\displaystyle\frac{1}{n}H\left(\bold{Y}_1|g(T_1,T_2),Z_1\right)\\
\nonumber &=&R^{*}_1+\displaystyle\frac{1}{n}H\left(\bold{Y}_1|\bold{\hat{X}}^{opt}_1,\bold{\hat{X}}^{opt}_2,Z_1\right)\\
&\leq&R^{*}_1+\displaystyle\frac{1}{n}H\left(\bold{Y}_1|\bold{\hat{X}}^{opt}_1,Z_1\right)
\end{eqnarray}
and 
\begin{eqnarray}
\nonumber R_2&=&R^{*}_2+\displaystyle\frac{1}{n}H\left(\bold{Y}_2|\bold{Y}_1,T_1,T_2,Z_1,Z_2\right)\\
\nonumber &\leq&R^{*}_2+\displaystyle\frac{1}{n}H\left(\bold{Y}_2|\bold{Y}_1,g(T_1,T_2),Z_1,Z_2\right)\\
\nonumber &=&R^{*}_2+\displaystyle\frac{1}{n}H\left(\bold{Y}_2|\bold{Y}_1,\bold{\hat{X}}^{opt}_1,\bold{\hat{X}}^{opt}_2,Z_1,Z_2\right)\\
&\leq&R^{*}_2+\displaystyle\frac{1}{n}H\left(\bold{Y}_2|\bold{\hat{X}}^{opt}_2,Z_2\right)
\end{eqnarray}
The upper bounds on $H(\bold{Y}_i|\bold{\hat{X}}^{opt}_i,{Z_i})$ can be obtained in the same way as in \cite{Ziv}. Notice that the Slepian-Wolf coding part in the proof requires long blocks of $(T_1,T_2,\bold{Y}_1,\bold{Y}_2)$. Now, since ${\cal{R}}(D_1,D_2)\subseteq {\cal{R}}^{*}(D_1,D_2)$, we can always find $R$ $(R^{*}_1,R^{*}_2)\in{\cal{R}}^{*}(D_1,D_2)$ such that $R^{*}_1+c\in \overline{{\cal{R}}}_1(D_1,D_2)$ (or higher and thus can be reduced to this range). Using the auxiliary scheme above, the rate pair $(R_1,R_2)=(R^{*}_1+c,R^{*}_2+c)$ can be achieved. Therefore, it can also be achieved by the dithered scheme, since $R_1\in \overline{{\cal{R}}}_1(D_1,D_2)$ (or higher), and in this range the regions of the auxiliary scheme and the dithered scheme coincide. Notice that any rate pair in ${\cal{R}}(D_1,D_2)$ can be achieved in practice by time-sharing the two edge points of ${\cal{R}}(D_1,D_2)$.
\end{proof}
\section{Revisiting the Upper Bound on $H\left(\bold{Y}|\bold{\hat{X}}^{opt},\bold{Z}\right)$}
In this section, we revisit the proof of \cite{Ziv} for the upper bound on $H\left(\bold{Y}|\bold{\hat{X}}^{opt},\bold{Z}\right)$. This is done for completeness and since we point and modify some of the steps in the next section.
The result of this section involves only one source $X$. The width of the quantization cell is denoted by $\Delta\triangleq 2\sqrt{3D}\Rightarrow D=\displaystyle{\Delta^2}/{12}$.

First, we show that for each coordinate $X_k$, $k\in\{1,\ldots,n\}$,
\begin{eqnarray}
\mathbb{E}\left[X_{k}-\hat{X}^{opt}_k+Z\right]=0.
\label{experror}
\end{eqnarray}
This follows from the following consideration:
\begin{eqnarray}
\nonumber\mathbb{E}\left[X_{k}-\hat{X}^{opt}_k+Z\right]&=&\mathbb{E}\left[X_{k}-\hat{X}^{opt}_k\right]+\mathbb{E}\left[Z\right]\\
&=&\mathbb{E}\left[X_{k}-\hat{X}^{opt}_k\right].
\end{eqnarray}
The distortion associated with $X_{k}$ is given by:
\begin{eqnarray}
\nonumber\mathbb{E}\left[\left(X_{k}-\hat{X}^{opt}_k\right)^2\right]&=&\text{Var}\{X_{k}-\hat{X}^{opt}_k)\}
+\left(\mathbb{E}\left[X_{k}-\hat{X}^{opt}_k\right]\right)^2\\
&\geq&\text{Var}\{X_{k}-\hat{X}^{opt}_k\}
\end{eqnarray}
where the inequality must be achieved by the optimal quantizer. Otherwise, we could add a constant to $\hat{X}^{opt}_k$ to obtain 
$\mathbb{E}\left[X_{k}-\hat{X}^{opt}_k\right]=0$ and thus smaller total distortion, in contradiction to the optimality of the quantizer.

We now rederive the upper bound on $H\left(\bold{Y}|\bold{X}^{opt},\bold{Z}\right)$. Using a method similar to \cite{Zamir_Feder1}, we show the following for the conditional entropy of each coordinate:
\begin{eqnarray}
\nonumber H\left(Y_k|{\hat{X}}_k^{opt},{Z}\right)&=&I\left({X}_k;{X}_k+{Z}|{\hat{X}}^{opt}_k\right)\\
&=&h\left({X}_k+{Z}|{\hat{X}}^{opt}_k\right)-h\left(Z\right)
\label{entropyvsmutual}
\end{eqnarray}
where the second equality follows since ${\hat{X}}^{opt}_k$ and $Z$ are independent.
By definition:
\begin{eqnarray}
\nonumber H\left( Y_k|\hat{X}^{opt}_k=q,Z\right)&=&\displaystyle\int_{-\frac{\Delta}{2}}^{\frac{\Delta}{2}} 
dz f_Z(z) H\left( Y_k|\hat{X}^{opt}_k=q,Z=z\right)\\
&=&\displaystyle \frac{1}{\Delta}\int_{-\frac{\Delta}{2}}^{\frac{\Delta}{2}}
 dz H\left( Y_k|\hat{X}^{opt}_k=q,Z=z\right)
\label{disentropy22}
\end{eqnarray}
Given $\left(\hat{X}^{opt}_k=q,Z=z\right)$, $Y_k$ is a discrete random variable taking values in $\{i\Delta\}_{i\in \mathbb{N}}$. Thus, 
\begin{eqnarray}
\nonumber H\left(Y_k|\hat{X}^{opt}_k=q,Z=z\right)&=&-\displaystyle\sum_{j\in \mathbb{N}}
P_{ Y_k|\hat{X}^{opt}_k,{Z}}(j\Delta|{q},{z})\\
&&\cdot \log  P_{ Y_k|\hat{X}^{opt}_k,{Z}}(j\Delta|{q},{z})
\end{eqnarray}
where $P_{ Y_k|\hat{X}^{opt}_k,{Z}}(\cdot|q,z)$ is the probability density function of $Y_k$ given $\hat{X}^{opt}_k$ and $Z$.
Calculating:
\begin{eqnarray}
\nonumber P_{Y_k|\hat{X}^{opt}_k,{Z}}(j\Delta|{q},{z})&=&
\Pr\{Y_k=j\Delta|{Z}={z},\hat{X}^{opt}_k={q}\}\\
\nonumber &=&\Delta\int_{(j-\frac{1}{2})\Delta-z}^{(j+\frac{1}{2})\Delta-z}dx \frac{1}{\Delta}f_{{X}|\hat{X}^{opt}_k}({x}|{q})\\
&=&\Delta\cdot f_{U_k|\hat{X}^{opt}_k}(j\Delta-{z}|{q})
\end{eqnarray}
where $f_{{X}|\hat{X}^{opt}_k}(\cdot|{q})$ is the probability density function of $X$ given $\hat{X}^{opt}_k$, $f_{U_k|\hat{X}^{opt}_k}(\cdot|{q})=f_{{X}|\hat{X}^{opt}_k}(\cdot|{q}) \ast f_{{Z}}(\cdot)$ is the probability density function of the continuous random variable $U_k\triangleq X_k+Z$ given $\hat{X}^{opt}_k$ and '$*$' denotes the convolution operation. Substituting in Eq. (\ref{disentropy22}), 
we have
\begin{eqnarray}
\nonumber H\left(Y_k|\hat{X}^{opt}_k={q},{Z}\right)&=&-\displaystyle \frac{1}{\Delta}
\int_{-\frac{\Delta}{2}}^{\frac{\Delta}{2}}
dz
\displaystyle\sum_{i\in\mathbb{N}}\Delta\cdot f_{{U_k}|\hat{X}^{opt}_k}(j\Delta-{z}|{q})\\
\nonumber &&\cdot\log\left(\Delta\cdot f_{{U_k}|\hat{X}^{opt}_k}(j\Delta-{z}|{q})\right)\\
\nonumber &=&-\displaystyle\int_{\mathbb{R}}du\cdot f_{{U_k}|\hat{X}^{opt}_k}({u}|{q})
\cdot\log\left(\Delta\cdot f_{{U_k}|\hat{X}^{opt}_k}({u}|{q})\right)\\
\nonumber &=&h\left(U_k|\hat{X}^{opt}_k={q}\right)-\log\Delta\\
\nonumber &=&h\left(U_k|\hat{X}^{opt}_k={q}\right)-h({Z})\\
\nonumber &=&h\left(U_k|\hat{X}^{opt}_k={q}\right)-h({Z}|X_k,\hat{X}^{opt}_k={q})\\
\nonumber &=&h\left(U_k|\hat{X}^{opt}_k={q}\right)-h({U_k}|X_k,\hat{X}^{opt}_k={q})\\
&=&I\left(X_k;X_k+Z|\hat{X}^{opt}_k={q}\right)
\end{eqnarray}
where in the fifth equality we used the independence of $X_k$ and $Z$ and in the sixth equality we used the fact that $U_k=X_k+Z$.
 We have
\begin{eqnarray}
\nonumber H\left( Y_k|\hat{X}^{opt}_k,{Z}\right)&=&\sum_{{q}\in {\cal{Q}}_{opt}}P_{\hat{X}^{opt}_k}({q}) H\left( Y_k|\hat{X}^{opt}_k={q},{Z}\right)\\
\nonumber &=&\sum_{{q}\in{\cal{Q}}_{opt}}P_{\hat{X}^{opt}_k}({q})I\left({X_k};{X_k}+{Z}|\hat{X}^{opt}_k={q}\right)\\
\nonumber &=& I\left({X}_k;{X}_k+{Z}|\hat{X}^{opt}_k\right)\\
\nonumber &=& h\left({X}_k+{Z}|\hat{X}^{opt}_k\right)-h({Z})\\
&=& h\left({X}_k-\hat{X}^{opt}_k+{Z}|\hat{X}^{opt}_k\right)-h({Z})
\end{eqnarray}
This completes the derivation of Eq. (\ref{entropyvsmutual}).
Now, we can upper bound $h\left({X}_k-\hat{X}^{opt}_k+{Z}|\hat{X}^{opt}_k\right)$ in the following way.
\begin{eqnarray}
\nonumber h\left({X}_k-\hat{X}^{opt}_k+{Z}|\hat{X}^{opt}_k\right)&=&\sum_{{q}\in {\cal{Q}}_{opt}}P_{\hat{X}^{opt}_k}({q}) h\left({X}-\hat{X}^{opt}_k+{Z}|\hat{X}^{opt}_k={q}\right)\\
\nonumber &\leq& \sum_{{q}\in {\cal{Q}}_{opt}}P_{\hat{X}^{opt}_k}({q})\\
\nonumber &&\cdot\displaystyle \frac{1}{2}\log \left( 2\pi e  \mathbb{E}\left[\left( X-\hat{X}^{opt}_k+Z\right)^2|\hat{X}^{opt}_k={q}\right]\right)\\
&\leq&\frac{1}{2}\log \left(2\pi e \displaystyle \mathbb{E}\left(X-\hat{X}^{opt}_k+{Z}\right)^2\right)
\label{entropypower}
\end{eqnarray}
where in the first inequality we upper bounded the differential entropy by using the maximum-entropy property of the Gaussian random variable and the second inequality is due to Jensen.
Using these results, we can upper bound $H\left(\bold{Y}|\bold{\hat{X}}^{opt},\bold{Z}\right)$.
\begin{eqnarray}
\nonumber H\left(\bold{Y}|\bold{\hat{X}}^{opt},\bold{Z}\right)&\leq&\displaystyle\sum_{k=1}^n H\left(Y_k|{\hat{X}}^{opt}_k,{Z}\right)\\
\nonumber &\leq&\displaystyle\sum_{k=1}^n \frac{1}{2}\log \left(2\pi e \displaystyle \mathbb{E}\left[\left(X_k-\hat{X}^{opt}_k+{Z}\right)^2\right]\right)
\\&&\nonumber-n h(Z)\\
\nonumber &\leq&\frac{n}{2}\log \left(2\pi e \displaystyle\frac{1}{n}\sum_{k=1}^n  \mathbb{E}\left[\left(X_k-{\hat{X}}^{opt}_k+{Z}\right)^2\right]\right)\\
\nonumber &&-n\log\Delta\\
\nonumber &=&\frac{n}{2}\log \left(2\pi e \displaystyle\frac{1}{n} 
\mathbb{E}\left\|\bold{X}-\bold{\hat{X}}^{opt}+\bold{Z}\right\|^2\right)\\
\nonumber &&-n\log\Delta\\
\nonumber &\leq&\frac{n}{2}\log \left(2\pi e 2D\right)-n\log\Delta\\
\nonumber &=&\frac{n}{2}\log \left(2\pi e 2D\right)-\displaystyle\frac{n}{2}\log(\Delta^2)\\
\nonumber &=&\frac{n}{2}\log \left(4\pi e D\right)-\displaystyle\frac{n}{2}\log\left(12D\right)\\
&=&\frac{n}{2}\log \left(\displaystyle\frac{\pi e}{3}\right)
\label{upper_bound_derivation}
\end{eqnarray}
where the third inequality is due to Jensen, and in the fourth we used the following.
\begin{eqnarray}
\nonumber\displaystyle\frac{1}{n}\mathbb{E}\left\|\bold{X}-\bold{\hat{X}}^{opt}+\bold{Z}\right\|^2
&=&\displaystyle\frac{1}{n}\mathbb{E}\left\|\bold{X}-\bold{\hat{X}}^{opt}\right\|^2
+\displaystyle\frac{1}{n}\mathbb{E}\left\|\bold{Z}\right\|^2\\
&\leq&2D
\end{eqnarray}
which stems from the independence of $X$ and $Z$. This completes the proof of the upper bound on $H\left(\bold{Y}|\bold{\hat{X}}^{opt},\bold{Z}\right)$. 
\section{Improving the Bounds by Adding an Estimation Stage}
The goal of this section is to enhance the results of Section 1 by improving the coding scheme described there. The idea is to decrease the distortion by adding an estimation stage at the decoder side. The new scheme works as follows. After producing $\bold{Y}_1,\bold{Y}_2$ and instead of just using them as outputs, the decoder uses them to estimate each one of the source vectors $(\bold{X}_1,\bold{X}_2)$. Since the sources and the quantization process (given $Z$) are memoryless, the estimation can be done on a symbol-by-symbol basis.

We begin with the following lemma:
\begin{lemma}
For the multi-terminal setting described in Section 1, we have ($i\in\{1,2\}$):
\begin{eqnarray}
\mathbb{E}[Y_i-{Z_i}]&=&\mathbb{E}[X_i]\\
\label{eY_variancew}
\mathbb{E}[(Y_i-{Z_i})^2]&=&\mathbb{E}[{X_i}^2]+D_i\\
\mathbb{E}[X_i(Y_i-{Z_i})]&=&\mathbb{E}[X_i^2]\\
\mathbb{E}\left[(Y_1-Z_1)(Y_2-Z_2)\right]&=&\mathbb{E}\left[X_1X_2\right]\\
\mathbb{E}\left[X_1(Y_2-Z_2)\right]&=&\mathbb{E}\left[X_1X_2\right]\\
\mathbb{E}\left[X_2(Y_1-Z_1)\right]&=&\mathbb{E}\left[X_1X_2\right]
\end{eqnarray}
\label{lem22}
\end{lemma}
Notice that the results above are true for each coordinate $k\in\{1,\ldots,n\}$.
The proof of Lemma \ref{lem22} is given in Appendix B.

The improved decoder described below requires the knowledge of the second-order statistics of the source. However, as Lemma \ref{lem22} shows, these statistics can be estimated from $\{\bold{Y}_i\}_{i=1}^2$, so universality can still be maintained.

The decoder of the multi-terminal setting uses the optimal linear estimator,  under the MMSE criterion, of $\{X_i\}_{i=1}^2$ given $\{{Q}(X_i+Z_i)-Z_i\}_{i=1}^2$.
The estimation error is calculated by using the results of Lemma \ref{lem22}. From now on, without loss of generality, we assume that $\mathbb{E}[X_1]=\mathbb{E}[X_2]=0$. The covariance matrix of $\underline{Y}\triangleq\left[Q(X_1+Z_1)-Z_1,Q(X_2+Z_2)-Z_2\right]$ is:
\begin{equation}
\Lambda=\left(
\begin{array}{ll}
\mathbb{E}[X^2_1]+D_1&\mathbb{E}[X_1X_2]\\
\mathbb{E}[X_1X_2]&\mathbb{E}[X^2_2]+D_2
\end{array}
\right)
\end{equation}
and the inverse matrix is:
\begin{equation}
\Lambda^{-1}=\displaystyle\frac{1}{|\Lambda|}\left(
\begin{array}{ll}
\mathbb{E}[X^2_2]+D_2&-\mathbb{E}[X_1X_2]\\
-\mathbb{E}[X_1X_2]&\mathbb{E}[X^2_1]+D_1
\end{array}
\right)
\end{equation}
The vector $\mathbb{E}\left[X_1\cdot\underline{Y}^{\dagger}\right]$ is given by:
\begin{equation}
\mathbb{E}\left[X_1\cdot\underline{Y}^{\dagger}\right]=\left(
\begin{array}{l}
\mathbb{E}[X^2_1]\\
\mathbb{E}[X_1X_2]
\end{array}
\right)
\end{equation}
It can be shown by direct calculation that
\begin{eqnarray}
\nonumber \Lambda^{-1}\mathbb{E}\left[X_1\cdot\underline{Y}^{\dagger}\right]
=\displaystyle\frac{1}{|\Lambda|}\left(\begin{array}{l}
|\Lambda|-D_1(\mathbb{E}[X^2_2]+D_2)\\
\mathbb{E}[X_1X_2]D_1
\end{array}\right)
\end{eqnarray}
Therefore, the optimal linear estimator of $X_1$ given the vector $\underline{Y}$ is:
\begin{eqnarray}
\hat{X}_1=\underline{Y}\cdot\displaystyle\frac{1}{|\Lambda|}\left(\begin{array}{l}
|\Lambda|-D_1(\mathbb{E}[X^2_2]+D_2)\\
\mathbb{E}[X_1X_2]D_1
\end{array}\right)
\end{eqnarray}
The error of the optimal linear estimator is given by:
\begin{eqnarray}
D_1^*=\mathbb{E}\left[X_1^2\right]-\mathbb{E}\left[\hat{X}_1^2\right]
\end{eqnarray}
It is shown in Appendix C that the estimation error takes the following form:
\begin{eqnarray}
\nonumber D_1^*=
D_1\frac{\mathbb{E}[X^2_1](\mathbb{E}[X^2_2]+D_2)-\mathbb{E}[X_1X_2]^2}{(\mathbb{E}[X^2_1]+D_1)(\mathbb{E}[X^2_2]+D_2)-
\mathbb{E}[X_1X_2]^2}
\label{est_error}
\end{eqnarray}
Remember that $D_1^*$ is the distortion of $X_1$ in the multi-terminal setting, where we add the above estimation stage after decoding $\left(\bold{Y}_1,\bold{Y}_2\right)$. It can be easily seen that the fraction in the brackets is less than $1$ and thus $D_1^*\leq D_1$ as desired. The same can be done, of course, for $X_2$. Since the distortion of $X_i$ in the improved scheme is $D_i^*$, we should compare the rate pair $(R_1,R_2)$ of this scheme, to the optimal rate pair $(R^{*}_1,R^{*}_2)$ which achieves $(D_1^*,D_2^*)$. This fact immediately improves on the results of Theorems 1 and 2. Revisiting the derivation of the upper bound for $H\left(\bold{Y}|\bold{\hat{X}}^{opt},\bold{Z}\right)$ in Eq. (\ref{upper_bound_derivation}), it can be shown that ($i\in\{1,2\}$):
\begin{eqnarray}
H\left(\bold{Y}_i|\bold{\hat{X}}^{opt}_i,\bold{Z}_i\right)\leq
\frac{n}{2}\log \left[\displaystyle\frac{\pi e}{6}\left(\frac{D^{*}_i}{D_i}+1\right)\right]
\label{upperbound222}
\end{eqnarray}
by using the following:
\begin{eqnarray}
\nonumber\displaystyle\frac{1}{n}\mathbb{E}\left\|\bold{X}_i-\bold{\hat{X}}_i^{opt}+\bold{Z}_i\right\|^2
&=&\displaystyle\frac{1}{n}\mathbb{E}\left\|\bold{X}_i-\bold{\hat{X}}_i^{opt}\right\|^2\\
\nonumber&+&\displaystyle\frac{1}{n}\mathbb{E}\left\|\bold{Z}_i\right\|^2
\\&\leq&D^*_i+D_i
\end{eqnarray}
Notice that when $X_1$ and $X_2$ are independent, $\mathbb{E}[X_1X_2]=0$ and we have
\begin{eqnarray}
H\left(\bold{Y}_i|\bold{\hat{X}}^{opt}_i,\bold{{Z_i}}\right)
\leq\frac{n}{2}\log \left[\displaystyle\frac{\pi e}{6}\left(2-\frac{D_i}
{\mathbb{E}[X^2_i]+D_i} \right)\right]
\label{upperbound333}
\end{eqnarray}
The maximum interesting value of $D_i^*$ is, of course, $\mathbb{E}[X^2_i]$. This value is obtained for $D_i\rightarrow\infty$. 
It is not hard to see that the range of the upper bound in (\ref{upperbound333}) is $[0.255,0.755]$ and that it is a decreasing function of $D_1$. For the high-SNR limit, i.e., $D_i\rightarrow 0$, it is well known that the redundancy is $0.255$ bits/sample (cf. \cite{Gish_Pierce}). We define ($i\in\{1,2\}$):
\begin{eqnarray}
c_i(D_1,D_2)=\frac{n}{2}\log \left[\displaystyle\frac{\pi e}{6}\left(\frac{D^{*}_i}{D_i}+1\right)\right]
\end{eqnarray}
We can now state Theorems 3 and 4. These theorems are obtained by applying the generalized upper bound of Eq. (\ref{upperbound222}), instead of Ziv's upper bound on $H\left(\bold{Y}_i|\bold{\hat{X}}^{opt}_i,\bold{{Z_i}}\right)$, in the proofs of Theorem 1 and 2.
\begin{theorem}
For any rate pair $(R^{*}_1,R^{*}_2)$ on the boundary of ${\cal{R}}^{*}(D_1,D_2)$ and any rate pair $(R_1,R_2)$ on the boundary of ${\cal{R}}(D^*_1,D^*_2)$, with $R_1\in \overline{{\cal{R}}}_1(D_1,D_2)$, we have
\begin{eqnarray}
R_1+R_2 \leq R^{*}_1+R^{*}_2+c_1(D_1,D_2)+c_2(D_1,D_2)
\end{eqnarray}
Moreover, for any $R^{*}_1\in \overline{{\cal{R}}}_1(D_1,D_2)$, there exists a rate pair $(R_1,R_2)\in{\cal{R}}(D^*_1,D^*_2)$ such that:
\begin{eqnarray}
\nonumber R_1&=&R^{*}_1\\
R_2&\leq& R^{*}_2+c_1(D_1,D_2)+c_2(D_1,D_2)
\end{eqnarray}
\label{the22}
\end{theorem}
\begin{theorem}
For any rate pair $(R_1,R_2)$ on the boundary of ${\cal{R}}(D^*_1,D^*_2)$, with $R_1\in \overline{{\cal{R}}}_1(D_1,D_2)$, there exists a rate pair $(R_1,R_2)\in{\cal{R}}(D^*_1,D^*_2)$ such that:
\begin{eqnarray}
\nonumber R_1&\leq& R^{*}_1+c_1(D_1,D_2)\\
          R_2&\leq& R^{*}_2+c_2(D_1,D_2)
\label{the22equationimproved}
\end{eqnarray}
\label{thestrongimproved}
\end{theorem}
\section*{Acknowledgment}
The authors are grateful to Prof. Rami Zamir for useful discussions.
\newpage
\section*{Appendix A - Universal Slepian-Wolf Coding}
\renewcommand{\theequation}{A.\arabic{equation}}
\setcounter{equation}{0}
In this appendix we describe the universal Slepian-Wolf decoder used in our coding scheme. 
The following results are similar to those of \cite{Csiszar}. For convenience, we omit the notation of the conditioning on the dither variables $Z_1$ and $Z_2$. The results below can be applied for any realization of these continuous variables. Remember that our coding scheme, unlike the scheme presented in \cite{Zamir_Berger}, requires only one realization of $Z_1$ and $Z_2$ in each round. 

We consider the Slepian-Wolf setting for two correlated memoryless sources $(Y_1,Y_2)\sim P_{Y_1,Y_2}$. We assume that $Y_1\in{\cal{Y}}_1$ and $Y_2\in{\cal{Y}}_2$, where ${\cal{Y}}_1$ and ${\cal{Y}}_2$ are finite alphabets. A $(2^{nR_1},2^{nR2}, n)$ source code is a block code of length $n$ consisting of two encoders ${f}_1$, ${f}_2$,
\begin{eqnarray}
{f}_1:{\cal{Y}}_1^n\rightarrow I_{M_1},&{f}_2:{\cal{Y}}_2^n\rightarrow I_{M_2}
\end{eqnarray}
and a decoder $g$
\begin{eqnarray}
g:I_{M_1} \times I_{M_2}\rightarrow {{\cal{Y}}}_1^n\times {{\cal{Y}}}_2^n.
\end{eqnarray}
where $M_j=2^{nR_j}$, $j=1,2$.
The probability of error of the code is defined as
\begin{eqnarray}
P_e(n)\triangleq\Pr\{g\left(f_1(\bold{Y}_1),f_2(\bold{Y}_2)\right)\neq\left(\bold{Y}_1,\bold{Y}_2\right)\}
\end{eqnarray}
We will prove the following result:
\begin{theorem}
	Let $(R_1,R_2)$ be given. Then, there exists a sequence of $(2^{nR_1} ,2^{nR2}, n)$ Slepian-Wolf source codes with probability of error $P_e(n)\rightarrow 0$ as $n\rightarrow \infty$ for every memoryless source that satisfies Eq. (\ref{slepian_wolf_eq}).
\end{theorem} 
\begin{proof}
Throughout the proof, the cardinality of a set ${\cal{A}}$ is denoted by $|{\cal{A}}|$. The empirical joint entropy $H_{\bold{y}_1,\bold{y}_2}(Y_1,Y_2)$ and the empirical conditional entropy $H_{\bold{y}_1,\bold{y}_2}(Y_1|Y_2)$  induced by the sequences $\bold{y}_1\in{\cal{Y}}^n_1$, $\bold{y}_2\in{\cal{Y}}^n_2$ are defined as
\begin{eqnarray}
H_{\bold{y}_1,\bold{y}_2}(Y_1,Y_2)&\triangleq&-\displaystyle \sum_{y_1\in{\cal{Y}}_1}\displaystyle \sum_{y_2\in{\cal{Y}}_2}P_{\bold{y}_1,\bold{y}_2}(y_1,y_2)\log P_{\bold{y}_1,\bold{y}_2}(y_1,y_2)\\
H_{\bold{y}_1,\bold{y}_2}(Y_1|Y_2)&\triangleq&-\displaystyle \sum_{y_1\in{\cal{Y}}_1}\displaystyle \sum_{y_2\in{\cal{Y}}_2}P_{\bold{y}_1,\bold{y}_2}(y_1,y_2)\log P_{\bold{y}_1,\bold{y}_2}(y_1|y_2)
\end{eqnarray}
where $P_{\bold{y}_1,\bold{y}_2}(y_1,y_2)$, $P_{\bold{y}_1,\bold{y}_2}(y_1|y_2)$ are the empirical joint and conditional distribution functions, respectively, induced by $\bold{y}_1$ and $\bold{y}_2$ (see \cite[Chap. 11]{Cover}).  

To prove the theorem, we use the following random-binning mechanism:
\begin{itemize}
	\item {\it{Codebook generation:}} Assign every $\bold{y}_1\in {\cal{Y}}_1^n$ to one of $2^{nR1}$ bins independently according to a uniform distribution on $\{1, 2,\ldots 2^{nR1}\}$. Similarly, randomly assign every  $\bold{y}_2\in {\cal{Y}}_2^n$ to one of $2^{nR2}$ bins.
	Reveal the assignments $f_1$ and $f_2$ to the encoders and the decoder.
	\item  {\it{Encoding:}} User $j$ sends the index of the bin to which $\bold{Y}_j$ belongs, $j=1,2$.
	\item  {\it{Decoding:}} Given the received index pair $\left(T_1=f_1(\bold{Y}_1), T_2=f_2(\bold{Y}_2)\right)$, the decoder uses the Minimum Joint Entropy (MJE) decoder: Choose the pair $(\bold{y}^{'}_1,\bold{y}^{'}_2): f_1(\bold{y}^{'}_1)=T_1, f_2(\bold{y}^{'}_2)=T_2$ which minimizes the empirical joint entropy induced by $(\bold{y}^{'}_1,\bold{y}^{'}_2)$, $H_{\bold{y}^{'}_1,\bold{y}^{'}_2}(Y_1,Y_2)$.
\end{itemize}
Define the following events:
\begin{eqnarray}
\nonumber E_0&=&\big\{ \left(\bold{Y}_1,\bold{Y}_2 \right)\notin A^n_{\epsilon}\big\}\\
\nonumber E_1&=&\big\{ \left(\bold{Y}_1,\bold{Y}_2\right)\in A^n_{\epsilon}\big\}\cap
\big\{\exists \bold{y}^{'}_1\neq \bold{Y}_1:f_1(\bold{y}^{'}_1)=T_1 \text{ and }  
H_{\bold{y}^{'}_1,\bold{Y}_2}(Y_1,Y_2)\leq H_{\bold{Y}_1,\bold{Y}_2}(Y_1,Y_2)\big\}\\
\nonumber E_2&=&\big\{\left( \bold{Y}_1,\bold{Y}_2 \right)\in A^n_{\epsilon} \big\}\cap\big\{\exists \bold{y}^{'}_2\neq \bold{Y}_2:f_2(\bold{y}^{'}_2)=T_2 \text{ and }  
H_{\bold{Y}_1,\bold{y}^{'}_2}(Y_1,Y_2)\leq H_{\bold{Y}_1,\bold{Y}_2}(Y_1,Y_2)\big\}\\
\nonumber E_{12}&=&\big\{\left( \bold{Y}_1,\bold{Y}_2\right)\in A^n_{\epsilon} \big\}\cap\big\{\exists \left(\bold{y}^{'}_1,\bold{y}^{'}_2\right): \bold{y}^{'}_1\neq \bold{Y}_1,\bold{y}^{'}_2\neq \bold{Y}_2,f_1(\bold{y}^{'}_1)=T_1,f_2(\bold{y}^{'}_2)=T_2\big.\\
&&\big. \text{ and } H_{\bold{y}^{'}_1,\bold{y}^{'}_2}(Y_1,Y_2)\leq H_{\bold{Y}_1,\bold{Y}_2}(Y_1,Y_2)\big\}
\end{eqnarray}
where $A^n_{\epsilon}$, $\epsilon>0$, is the strongly typical set with respect to the source $P_{Y_1,Y_2}$ (see \cite[Eq. 10.107]{Cover}).
Remember that $\bold{Y}_1$, $\bold{Y}_2$, $f_1$ and $f_2$ are random.
Obviously,
\begin{eqnarray}
H_{\bold{y}^{'}_1,\bold{y}_2}(Y_1,Y_2)\leq H_{\bold{y}_1,\bold{y}_2}(Y_1,Y_2) \Leftrightarrow 
H_{\bold{y}^{'}_1,\bold{y}_2}(Y_1|Y_2)\leq H_{\bold{y}_1,\bold{y}_2}(Y_1|Y_2) 
\end{eqnarray}
where $H_{\bold{y}^{'}_1,\bold{y}_2}(Y_1|Y_2)$, $H_{\bold{y}_1,\bold{y}_2}(Y_1|Y_2)$ are the empirical conditional entropies induced by $({\bold{y}^{'}_1,\bold{y}_2})$ and $({\bold{y}_1,\bold{y}_2})$, respectively.
We have an error if there is another pair of sequences in the same bin such that the empirical joint entropy induced by this pair is smaller than the empirical joint entropy induced by $\left(\bold{Y}_1,\bold{Y}_2\right)$. Hence,
\begin{eqnarray}
\nonumber \bar{P}_e(n)&\leq& \Pr\left\{E_0\cup E_1 \cup E_2 \cup E_{12} \right\}\\
&\leq&\Pr\left\{ E_0\right\}+\Pr\left\{ E_1\right\}+\Pr\left\{ E_2\right\}+\Pr\left\{ E_{12}\right\}
\end{eqnarray}
where $\bar{P}_e(n)\triangleq  {\mathbb{E}} [{P}_e(n)]$ is the expected probability of error where the expectation is taken with respect to the random choice of the code. The first inequality follows from the fact that we treat $E_0$ as error event and the second inequality is due to the union bound.
We first consider $E_0$. By the asymptotic equipartition property (AEP), $\Pr\left\{ E_0\right\}\rightarrow 0$ and hence for $n$ sufficiently large, $\Pr\left\{ E_0\right\}<\epsilon$. To bound $\Pr(E_1)$, we have
\begin{eqnarray}
\nonumber\Pr(E_1)&=&\displaystyle\sum_{\left( \bold{y}_1,\bold{y}_2\right)\in A_{\epsilon}^n}
P(\bold{y}_1,\bold{y}_2)\\
\nonumber&&\cdot\Pr\big\{\exists \bold{y}^{'}_1\neq \bold{y}_1:f_1(\bold{y}^{'}_1)=f_1(\bold{y}_1)\text{ and }  
H_{\bold{y}^{'}_1,\bold{y}_2}(Y_1|Y_2)\leq H_{\bold{y}_1,\bold{y}_2}(Y_1|Y_2)\big\}\\
\nonumber&\leq&\displaystyle\sum_{\left( \bold{y}_1,\bold{y}_2\right)\in A_{\epsilon}^n}
P(\bold{y}_1,\bold{y}_2)\displaystyle\sum_{\bold{y}^{'}_1\in B(\bold{y}_1,\bold{y}_2)}\Pr\{   f_1(\bold{y}^{'}_1)=  f_1(\bold{y}_1)\}\\
&=&\displaystyle\sum_{\left( \bold{y}_1,\bold{y}_2\right)\in A_{\epsilon}^n}
P(\bold{y}_1,\bold{y}_2) 
2^{-nR_1} \left|B(\bold{y}_1,\bold{y}_2)\right|
\end{eqnarray}
where the set $B(\bold{y}_1,\bold{y}_2)$ is defined as 
\begin{eqnarray}
B(\bold{y}_1,\bold{y}_2)\triangleq\big\{\bold{y}^{'}_1: H_{\bold{y}^{'}_1,\bold{y}_2}(Y_1|Y_2)\leq H_{\bold{y}_1,\bold{y}_2}(Y_1|Y_2)\big\}
\end{eqnarray}
and the last equality simply follows from the definition of the random-binning coding scheme.
Using the method of types (see \cite[Chap. 10-11]{Cover}), we have
\begin{eqnarray}
\nonumber|B(\bold{y}_1,\bold{y}_2)|&=&\displaystyle\sum_{\bold{y}^{'}_1\in B(\bold{y}_1,\bold{y}_2)} 1\\
\nonumber&=&\displaystyle\sum_{{V_{\bold{y}^{'}_1|\bold{y}_2}}\subseteq B(\bold{y}_1,\bold{y}_2)}\left|V_{\bold{y}^{'}_1|\bold{y}_2}\right|\\
\nonumber&\leq&\displaystyle\sum_{V_{\bold{y}^{'}_1|\bold{y}_2}\subseteq B(\bold{y}_1,\bold{y}_2)}
2^{n\big( H_{\bold{y}^{'}_1,\bold{y}_2}(Y_1|Y_2)+\epsilon\big)}\\
\nonumber&\leq&\displaystyle\sum_{V_{\bold{y}^{'}_1|\bold{y}_2}\subseteq B(\bold{y}_1,\bold{y}_2)}
2^{n\big( H_{\bold{y}_1,\bold{y}_2}(Y_1|Y_2)+\epsilon\big)}\\
\nonumber&\leq&(n+1)^{|{\cal{Y}}_1||{\cal{Y}}_2|} 2^{n\big( H_{\bold{y}_1,\bold{y}_2}(Y_1|Y_2)+\epsilon\big)}\\
&\leq&(n+1)^{|{\cal{Y}}_1||{\cal{Y}}_2|}2^{n\big( H(Y_1|Y_2)+2\epsilon\big)}
\label{appA1}
\end{eqnarray}
where $V_{\bold{y}^{'}_1|\bold{y}_2}$ is the conditional type of $\bold{y}^{'}_1$ given $\bold{y}_2$ (see \cite[Chap. 10]{Cover}). The second equality follows from the fact that the event $\bold{y}^{'}_1\in B(\bold{y}_1,\bold{y}_2)$ depends only on the type $V_{\bold{y}^{'}_1|\bold{y}_2}$. In the first inequality, we used the known upper bound on the size of the conditional type. The second inequality stems from the definition of $B(\bold{y}_1,\bold{y}_2)$. In the third inequality we used a known upper bound on the number of conditional types. The last inequality follows since (see \cite[Chap. 10]{Cover})
\begin{eqnarray}
(\bold{y}_1,\bold{y}_2)\in A_{\epsilon}^n &\Rightarrow& H{\bold{y}_1,\bold{y}_2}(Y_1|Y_2)\leq H(Y_1|Y_2)+\epsilon.
\label{appA2}
\end{eqnarray}
Therefore, we have
\begin{eqnarray}
\nonumber\Pr(E_1)&\leq&\displaystyle\sum_{\left( \bold{y}_1,\bold{y}_2\right)\in A_{\epsilon}^n}
P(\bold{y}_1,\bold{y}_2) 
2^{-nR_1} \left|B(\bold{y}_1,\bold{y}_2)\right|\\
&\leq&(n+1)^{|{\cal{Y}}_1||{\cal{Y}}_2|} 2^{-nR_1}2^{n(H(Y_1|Y_2)+2\epsilon)}
\end{eqnarray}
where in the second inequality we used Eq. (\ref{appA1}).
Similarly, it can be shown that
\begin{eqnarray}
\Pr(E_2)
&\leq& (n+1)^{|{\cal{Y}}_1||{\cal{Y}}_2|} 2^{-nR_2}\cdot 2^{n(H(Y_2|Y_1)+2\epsilon)}
\end{eqnarray}
and
\begin{eqnarray}
\Pr(E_{12})
&\leq& (n+1)^{|{\cal{Y}}_1||{\cal{Y}}_2|} 2^{-n(R_1+R_2)}\cdot 2^{n(H(Y_1,Y_2)+2\epsilon)}
\end{eqnarray}
Hence, taking $R_1>H(Y_1|Y_2)+2\epsilon$,  $R_2>H(Y_2|Y_1)+2\epsilon$ and $R_1 + R_2 > H(Y_1,Y_2)+2\epsilon$, we have $P(E_1) <\epsilon$, $P(E_2) <\epsilon$ and $P(E_{12}) <\epsilon$ for sufficiently large $n$. Since $\bar{P}_e(n)\leq 4\epsilon$, there exists at least one universal code $(f_1^{*},f_2^{*},g^{*})$ with 
$P_e(n)\leq 4\epsilon$. Thus, we can construct a sequence of universal codes with $P_e(n)\rightarrow 0$, and
the proof of achievability is complete.
\end{proof}
{\it{Remark}}. It can be shown that the universal decoder presented in the proof above also achieves the optimal error exponent. 
\section*{Appendix B - Proof of Lemma \ref{lem22}}
\renewcommand{\theequation}{B.\arabic{equation}}
\setcounter{equation}{0}
We now prove Lemma \ref{lem22}. We first show that the random vector $\left(Y_1-Z_1,Y_2-Z_2\right)$ is equivalent to the random vector $\left(X_1+N_1,X_2+N_2\right)$ where $N_1$, $N_2$ are independent of $X_1$, $X_2$ and of each other and $N_i \sim {\cal{U}}[-\sqrt{3D_i},\sqrt{3D_i}]$, $i\in\{1,2\}$. Therefore, the dithered quantization process can be viewed as passing $X_1$, $X_2$ through independent noisy memoryless channels $\hat{X}_1=X_1+N_1$ and $\hat{X}_2=X_2+N_2$, respectively. We start with the following conditional probability distribution.
\begin{eqnarray}
f_{N_1,N_2|X_1,X_2}\left(N_1,N_2|X_1,X_2\right)=f\left(N_1|X_1\right)f\left(N_2|X_2\right)
\end{eqnarray}
where we have defined $N_1\triangleq Y_1-Z_1-X_1$, $N_2\triangleq Y_2-Z_2-X_2$. The equality stems from the fact that $\left(Y_1-Z_1-X_1\right)$ is independent of $X_2$ given $X_1$ and  $\left(Y_2-Z_2-X_2\right)$ is independent of $X_1$ given $X_2$, since $(Z_1,Z_2)$ are independent of $(X_1,X_2)$. In addition, it can be easily seen that for every value of $X_i$, $N_i$ is uniformly distributed over $[-\sqrt{3D_i},\sqrt{3D_i}]$. Therefore, $N_i$ is independent of $X_i$ and we have
\begin{eqnarray}
f_{N_1,N_2|X_1,X_2}\left(N_1,N_2|X_1,X_2\right)=f\left(N_1\right)f\left(N_2\right)
\end{eqnarray}
Lemma \ref{lem22} follows directly from this result:
\begin{equation}
\begin{array}{lllll}
\mathbb{E}[Y_i-{Z_i}]&=&\mathbb{E}[X_i+N_i]&=&\mathbb{E}[X_i]\\
\mathbb{E}[(Y_i-{Z_i})^2]&=&\mathbb{E}[(X_i+N_i)^2]&=&\mathbb{E}[{X_i}^2]+D_i\\
\mathbb{E}[X_i(Y_i-{Z_i})]&=&\mathbb{E}[X_i(X_i+N_i)]&=&\mathbb{E}[X_i^2]\\
\mathbb{E}\left[(Y_1-Z_1)(Y_2-Z_2)\right]&=&\mathbb{E}\left[(X_1+N_1)(X_2+N_2)\right]&=&\mathbb{E}\left[X_1X_2\right]\\
\mathbb{E}\left[X_1(Y_2-Z_2)\right]&=&\nonumber\mathbb{E}\left[X_1(X_2+N_2)\right]&=&\mathbb{E}\left[X_1X_2\right]\\
\mathbb{E}\left[X_2(Y_1-Z_1)\right]&=&\mathbb{E}\left[X_2(X_1+N_1)\right]&=&\mathbb{E}\left[X_1X_2\right]
\end{array}
\end{equation}
\section*{Appendix C - Calculation of the Estimation Error}
\renewcommand{\theequation}{C.\arabic{equation}}
\setcounter{equation}{0}
In this appendix we calculate the estimation error given in Eq. (\ref{est_error}).
The optimal linear estimator of $X_1$ given the vector $\underline{Y}$ is:
\begin{equation}
\begin{array}{lll}
\hat{X}_1&=&\underline{Y}\cdot\displaystyle\frac{1}{|\Lambda|}\left(\begin{array}{l}
|\Lambda|-D_1(\mathbb{E}[X^2_2]+D_2)\\
\mathbb{E}[X_1X_2]D_1
\end{array}\right)
\end{array}
\end{equation}
where 
\begin{eqnarray}
{|\Lambda|}=(\mathbb{E}[X^2_1]+D_1)(\mathbb{E}[X^2_2]+D_2)-
\mathbb{E}[X_1X_2]^2
\end{eqnarray}
The error of the optimal linear estimator is given by:
\begin{eqnarray}
D_1^*&=&\mathbb{E}\left[X_1^2\right]-\mathbb{E}\left[\hat{X}_1^2\right]
\end{eqnarray}
Calculating the second term:
\begin{eqnarray}
\nonumber |\Lambda|^2\mathbb{E}\left[\hat{X}_1^2\right]&=&\left(|\Lambda|-D_1\left(\mathbb{E}[X^2_2]+D_2\right)\right)^2\mathbb{E}\left[\left(Y_1-Z_1
\right)^2\right]\\
\nonumber&&+\mathbb{E}[X_1X_2]^2D_1^2\mathbb{E}\left[\left(Y_2-Z_2\right)^2\right]\\
\nonumber&&+2\left(|\Lambda|-D_1\left(\mathbb{E}[X^2_2]+D_2\right)\right)\mathbb{E}[X_1X_2]D_1\mathbb{E}\left[\left(Y_1-Z_1\right)\left(Y_2-Z_2\right)\right]\\
\nonumber&=&\left(|\Lambda|-D_1\left(\mathbb{E}[X^2_2]+D_2\right)\right)^2\left(\mathbb{E}[X^2_1]+D_1\right)\\
\nonumber&&+\mathbb{E}[X_1X_2]^2D_1^2\left(\mathbb{E}[X^2_2]+D_2\right)\\
\nonumber&&+2\left(|\Lambda|-D_1\left(\mathbb{E}[X^2_2]+D_2\right)\right)\mathbb{E}[X_1X_2]^2D_1\\
\nonumber&=&\left(|\Lambda|-D_1\left(\mathbb{E}[X^2_2]+D_2\right)\right)^2\left(\mathbb{E}[X^2_1]+D_1\right)\\
\nonumber&&+\mathbb{E}[X_1X_2]^2D_1\left(D_1\left(\mathbb{E}[X^2_2]+D_2\right)+2\left(|\Lambda|-D_1\left(\mathbb{E}[X^2_2]+D_2\right)\right)
\right)\\
\nonumber&=&\left(|\Lambda|-D_1\left(\mathbb{E}[X^2_2]+D_2\right)\right)^2\left(\mathbb{E}[X^2_1]+D_1\right)\\
\nonumber&&+\mathbb{E}[X_1X_2]^2D_1\left(2|\Lambda|-D_1\left(\mathbb{E}[X^2_2]+D_2\right)\right)\\
\nonumber &=&\left(|\Lambda|-D_1\left(\mathbb{E}[X^2_2]+D_2\right)\right)
\nonumber\left(\left(|\Lambda|-D_1\left(\mathbb{E}[X^2_2]+D_2\right)\right)\left(\mathbb{E}[X^2_1]+D_1\right)+\mathbb{E}[X_1X_2]^2D_1\right) \\
\nonumber&&+|\Lambda|\mathbb{E}[X_1X_2]^2D_1\\
\nonumber&=&\left(|\Lambda|-D_1\left(\mathbb{E}[X^2_2]+D_2\right)\right)
\nonumber\left(|\Lambda|\left(\mathbb{E}[X^2_1]+D_1\right)-D_1|\Lambda|\right) \\
\nonumber&&+|\Lambda|\mathbb{E}[X_1X_2]^2D_1\\
\nonumber&=&\left(|\Lambda|-D_1\left(\mathbb{E}[X^2_2]+D_2\right)\right)
\nonumber|\Lambda|\mathbb{E}[X^2_1] \\
\nonumber&&+|\Lambda|\mathbb{E}[X_1X_2]^2D_1\\
\nonumber&=&|\Lambda|^2\mathbb{E}[X^2_1]
\nonumber+|\Lambda|D_1\left(\mathbb{E}[X_1X_2]^2-\mathbb{E}[X^2_1]\left(\mathbb{E}[X^2_2]+D_2\right)\right)\\
\nonumber&=&|\Lambda|^2\mathbb{E}[X^2_1]
\nonumber-|\Lambda|D_1\left(\mathbb{E}[X^2_1]\left(\mathbb{E}[X^2_2]+D_2\right)-\mathbb{E}[X_1X_2]^2\right)\\
&=&|\Lambda|^2\mathbb{E}[X^2_1]
-|\Lambda|D_1\left(|\Lambda|-D_1\left(\mathbb{E}[X^2_2]+D_2\right)\right)
\end{eqnarray}
where in the second equality we used the results of Lemma \ref{lem22}.
Therefore, we have
\begin{eqnarray}
\nonumber D_1^*&=&\frac{|\Lambda|D_1\left(|\Lambda|-D_1(\mathbb{E}[X^2_2]+D_2))\right)}{|\Lambda|^2}\\
\nonumber&=&D_1\left(1-\frac{D_1(\mathbb{E}[X^2_2]+D_2)}{|\Lambda|}\right)\\
\nonumber&=&D_1\left(1-\frac{D_1(\mathbb{E}[X^2_2]+D_2)}{(\mathbb{E}[X^2_1]+D_1)(\mathbb{E}[X^2_2]+D_2)-
\mathbb{E}[X_1X_2]^2}\right)\\
&=&D_1\frac{\mathbb{E}[X^2_1](\mathbb{E}[X^2_2]+D_2)-\mathbb{E}[X_1X_2]^2}{(\mathbb{E}[X^2_1]+D_1)(\mathbb{E}[X^2_2]+D_2)-
\mathbb{E}[X_1X_2]^2}
\end{eqnarray}


\clearpage
\begin{thebibliography}{1}
\bibitem{Slepian_Wolf}
D.~Slepian and J.~K.~Wolf, ``Noiseless coding of correlated information
sources,'' {\it IEEE Trans.\ Inform.\ Theory}, vol.\ 19, pp.\ 471--480, Jul.
1973.
\bibitem{Wyner_Ziv}
A.~D.~Wyner and J.~Ziv, ``The rate-distortion function for source coding
with side information at the decoder,'' {\it IEEE Trans. Inform. Theory}, vol.\ 22, pp.\ 1--10, Jan. 1976.
\bibitem{Wyner}
A.~D.~Wyner, ``The rate-distortion function for source coding with side
information at the decoder-II: General sources,'' {\it Inform. Contr.}, vol.\
38, pp.\ 60--80, 1978.
\bibitem{Ahlswede_Korner}
R.~Ahlswede and J.~K\"{o}rner, ``Source coding with side information and
a converse for degraded broadcast channels,'' {\it IEEE Trans. Inform. Theory}, vol.\ 21, no.\ 6, pp.\ 629--637, Nov. 1975.
\bibitem{Berger_Yeung}
T.~Berger and R.~W.~Yeung, ``Multiterminal source encoding with
one distortion criterion,'' {\it IEEE Trans.\ Inform.\ Theory}, vol.\ 35, no.\ 2, pp.\ 228--236, Mar. 1989.
\bibitem{Zamir_Berger}
R.~Zamir and T.~Berger, ``Multiterminal source coding with high resolution,''
{\it IEEE Trans. Inf. Theory}, vol.\ 45, pp.\ 106--117, Jan. 1999.
\bibitem{Wagner_Anantharam}
A.~Wagner and V.~Anantharam, ``An improved outer bound for
multiterminal source coding,'' {\it IEEE Trans. Inf. Theory}, vol.\ 54, no.\ 5, pp.\ 1919--1937, May 2008.
\bibitem{Wagner_Tavildar_Viswanath}
A.~Wagner, S.~Tavildar, and P.~Viswanath, ``Rate region of the quadratic
gaussian two-encoder source-coding problem,'' {\it IEEE Trans. Inf. Theory}, vol.\ 54, no.\ 5, pp.\ 1938--1961, May 2008.
\bibitem{Berger_Tung}
T.~Berger and S.~Y.~Tung, ``Encoding of correlated analog sources,''
in {\it Proc. IEEE--USSR Joint Workshop on Information Theory}, pp.\ 7--10, 1975.
\bibitem{Courtade_Weissman}
T.~A.~Courtade and T.~Weissman, ``Multiterminal source coding under
logarithmic loss,'' {\it IEEE Trans. Inform. Theory}, vol.\ 60, pp.\ 740--761, Jan. 2014.
\bibitem{Kaspi_Merhav}
Y.~Kaspi and N.~Merhav, ``Zero-delay and causal single-user and multi-user lossy source coding with decoder side information,'' {\it IEEE Trans. Inf. Theory}, vol.\ 60, no.\ 11, pp. 6931–-6942, Nov. 2014.
\bibitem{Ziv}
J.~Ziv, ``On universal quantization,'' {\it IEEE Trans. Inform. Theory}, vol.\ 31, pp.\
344--347, May 1985.
\bibitem{Zamir_Feder1}
R.~Zamir and M.~Feder, ``On universal quantization by randomized
uniform/lattice quantizer,'' {\it IEEE Trans. Inform. Theory}, vol.\ 38, pp.\
428--436, Mar. 1992.
\bibitem{Zamir_Feder2}
R.~Zamir and M.~Feder, ``Information rates of pre/post filtered dithered quantizers,'' {\it IEEE
Trans. Inform. Theory}, vol.\ 42, pp.\ 1340--1353, Sept. 1996.
\bibitem{Kieffer}
J.~Kieffer, ``A unified approach to weak universal source coding,'' {\it IEEE
	Trans. Inform. Theory}, vol.\ 24, pp.\ 674--682, Nov. 1978.
\bibitem{Gish_Pierce}
H.~Gish and N.~J.~Pierce, ``Asymptotically efficient quantization,''
{\it IEEE Trans. Inform. Theory}, vol.\ 14, pp.\ 676--683, Sept. 1968.
\bibitem{Csiszar}
I.~Csiszar and J.~K\"{o}rner, ``Towards a general theory of source networks,''
{\it IEEE Trans. Inform. Theory}, vol.\ 26, pp.\ 155--165, Mar. 1980.
\bibitem{Cover}
T. M.~Cover and J. A.~Thomas, ``Elements of information theory,'', 
{\it John Wiley} \& {\it Sons, 2nd Edition}.
\end{thebibliography}
\end{document}